\documentclass[11pt]{article}%
\usepackage{amsfonts}
\usepackage{amsmath}
\usepackage{amssymb}
\usepackage{graphicx}%
\providecommand{\U}[1]{\protect\rule{.1in}{.1in}}
\newtheorem{theorem}{Theorem}
\newtheorem{acknowledgement}[theorem]{Acknowledgement}

\newtheorem{corollary}[theorem]{Corollary}

\newtheorem{definition}[theorem]{Definition}

\newtheorem{lemma}[theorem]{Lemma}

\newtheorem{proposition}[theorem]{Proposition}
\newtheorem{remark}[theorem]{Remark}

\newenvironment{proof}[1][Proof]{\noindent\textbf{#1.} }{\ \rule{0.5em}{0.5em}}
\begin{document}

\title{Quantum Monads in Phase Space and Related Toeplitz Operators \ }
\author{Maurice de Gosson\thanks{mauricee.de.gosson@univie.ac.at}\\Austrian Academy of Sciences\\Acoustics Research Institute\\1010, Vienna, AUSTRIA\\and\\University of Vienna\\Faculty of Mathematics (NuHAG)\\1090 Vienna, AUSTRIA}
\maketitle

\begin{abstract}
In earlier work, we introduced quantum blobs as minimum-uncertainty symplectic
ellipsoids in phase space. These objects may be viewed as geometric monads in
the Leibnizian sense, representing the elementary units of phase-space
structure consistent with the uncertainty principle. We establish a one-to-one
correspondence between such monads and generalized coherent states,
represented by arbitrary non-degenerate Gaussian wave functions in
configuration space. To each of these states, we associate a classs of
Toeplitz operators that extends the standard anti-Wick quantization scheme.
The mathematical and physical properties of these operators are analyzed,
allowing for a generalized definition of density matrices within the
phase-space formulation of quantum mechanics.

\end{abstract}

Quantum mechanics.
\title{Quantum Monads in Phase Space and Related Toeplitz Operators \ }
\author{Maurice de Gosson\thanks{mauricee.de.gosson@univie.ac.at}\\Austrian Academy of Sciences\\Acoustics Research Institute\\1010, Vienna, AUSTRIA\\and\\University of Vienna\\Faculty of Mathematics (NuHAG)\\1090 Vienna, AUSTRIA}
\maketitle

\begin{abstract}
In earlier work, we introduced quantum blobs as minimum-uncertainty symplectic
ellipsoids in phase space. These objects may be viewed as geometric monads in
the Leibnizian sense, representing the elementary units of phase-space
structure consistent with the uncertainty principle. We establish a one-to-one
correspondence between such monads and generalized coherent states,
represented by arbitrary non-degenerate Gaussian wave functions in
configuration space. To each of these states, we associate a classs of
Toeplitz operators that extends the standard anti-Wick quantization scheme.
The mathematical and physical properties of these operators are analyzed,
allowing for a generalized definition of density matrices within the
phase-space formulation of quantum mechanics.

\end{abstract}
\tableofcontents

\section{Introduction and Preamble}

echanics, in contrast to classical mechanics, exhibits several distinctive
features. It is founded on the uncertainty principle, and its dynamics are
governed by non-commuting operators---two aspects that are intimately
connected. Another fundamental conceptual difference lies in the nature of
their dynamical descriptions: while classical mechanics, in its Hamiltonian
formulation, describes the evolution of points in phase space through
continuous flows, quantum mechanics represents dynamics by unitary evolution
operators acting on wavefunctions.

In the present work, we discuss the central role played by the notion of a
quantum blob in the operator approach to quantization. Quantum blobs,
introduced in our earlier work
$<$%
cite%
$>$%
blobs, go09%
$<$%
/cite%
$>$%
, are minimum-uncertainty cells in phase space; technically, they are
symplectic balls whose radius allows one to express the strong
Robertson--Schr\"{o}dinger uncertainty principle in purely geometrical terms,
involving the scale $\sqrt{\hbar}$. In this sense, they may be regarded as
phase-space quantum monads, in the spirit of Leibniz.

It turns out that quantum blobs are in bijection with generalized coherent
states, represented by the most general non-degenerate Gaussian wave packets
on configuration space. This remarkable property allows one to canonically
associate to each quantum blob a certain operator belonging to a subclass of
the Toeplitz or anti-Wick operators, which are widely used both in quantum
mechanics and in time--frequency analysis---two disciplines whose mutual
developments often stimulate each other.

These operators are obtained from the usual Weyl operators by smoothing their
symbols (the \textquotedblleft classical observables\textquotedblright) with a
Gaussian determined by the corresponding quantum blob. Such operators possess
two decisive advantages: positivity and physical interpretability, making them
particularly suitable for representing density operators and quantum
probabilities. Moreover, Toeplitz operators with bounded symbols are bounded
on the Hilbert space of square-integrable functions, which ensures their
analytical stability.

\subsection{Structure and description of this work}

\begin{itemize}
\item In Section \ref{sec1} we review the material needed from the
Weyl--Wigner--Moyal formalism, named in tribute to the mathematicians and
physicists who founded the modern symplectic and harmonic-analytic approaches
to quantum mechanics. Our presentation is rigorous and perhaps somewhat
unconventional from a physicist's standpoint (see, for instance, our
definition of the Wigner transform in terms of reflection operators). However,
this approach has the advantage of clarifying several subtle aspects of the
theory. A further novelty is the introduction of a particularly useful Banach
space, the Feichtinger algebra, which serves as a substitute for the usual
Schwartz space of rapidly decreasing functions. This topic is rarely treated
in physics-oriented literature. We discuss the fundamental relation between
quantum blobs and the generalized coherent states (called hereafter
"Gaussians"). Both can be mathematically identified using some techniques from
harmonic analysis. These results are not new, since they have been described
for instance in \cite{blobs} and \cite{Birkbis}, but they are presented i a
more concise and perhaps simpler way. We put a special emphasis on the
creation between quantum blobs and the strong version of the uncertainty
principle, as was described in our earlier paper \cite{go09} and further
developed on \cite{golu09}.

\item In Section \ref{sec2} we review the topicsw4e will need of what is
called the "Weyl--Wigner--Moyal formalism" in tribute to the those
mathematicians and physicists at the origin of the modern symplectic harmonic
approach to quantum mechanics. Our presenting is rigorous and perhaps somewhat
unusual to many physicists (see for instance our definition of the Wigner
transform in terms of reflection operators), but this has advantages in a
sense that it clarifies many aspects of the theory. Another novelty is the
introduction of a very useful Banach algebra, the Feichtinger algebra, which
is substitute for the ordinary space of Schwartz functions with rapid
decrease. This topic is usually not addressed in physics texts.

\item In \ Section \ref{sec3} we introduce the concept of Weyl--Heisenberg (or
Gabor) multipliers, familiar in time--frequency analysis but little known
among physicists. These operators can be viewed as discretized versions of
Weyl operators, defined in terms of frames, which generalize the standard
notion of a basis. They provide an effective tool for studying mixed quantum
states through their associated den

\item In Section \ref{sec4} we extend Weyl--Heisenberg multipliers to the
continuous case, leading to the notion of Toeplitz operators (which reduce to
anti--Wick operators when the window function is the standard Gaussian
coherent state). Essentially, Toeplitz operators are Weyl operators whose
symbols (observables) have been smoothed by convolution with a well-localized
regularizing function. When this regularizing function is chosen to be a
generalized Gaussian---that is, a quantum blob---the resulting operators
exhibit remarkable properties: in particular, they are asymptotically close to
the corresponding Weyl operators in the semiclassical limit.
\end{itemize}

\subsection{Preamble}

\subsubsection*{Points in classical and quantum mechanics}

The backbone of classical mechanics is phase space. It is a mathematical
object consists of pairs $(x,p)$ where $x=(x_{1},...,x_{n})$ and
$p=(p_{1},...,p_{n})$ where the $x_{j}$ are real numbers representing position
coordinates the $p_{j}$ momentum coordinates; we will denote the set of all
position vectors by $\mathbb{R}_{x}^{n}$ (or simply $\mathbb{R}^{n}$) and the
set of momentum vectors by $\mathbb{R}_{p}^{n}$ (or (or only $\mathbb{R}^{n}%
$). The phase space is then, by definition the product $\mathbb{R}_{x}%
^{n}\times\mathbb{R}_{p}^{n}$. In this context, $\mathbb{R}_{p}^{n}$ is often
identified wit the dual\ space\ $\ (\mathbb{R}_{x}^{n})^{\ast}$ of the
position space of; this point of view has two advantages: the first that it
allows to define in a natural way the standard symplectic structure on
$\mathbb{R}_{x}^{n}\times\mathbb{R}_{p}^{n}$ by
\[
\sigma(x,p;x^{\prime},p^{\prime}=p^{\prime}(x)-p(x^{\prime}).
\]
The second advantage of the identification $\mathbb{R}_{x}^{n}\times
\mathbb{R}_{p}^{n}\equiv\mathbb{R}_{x}^{n}\times(\mathbb{R}_{x}^{n})^{\ast}$
is that it is consistent with the view that if the position space
$\mathbb{R}_{x}^{n}$ is replace with a manifold $M$ then the phase space is be
the cotangent bundle $T^{\ast}M$ which is a natural generalization of the flat
case, the symplectic form being in this case the canonical two-form
\[
dp\wedge dx=dp_{1}\wedge dx_{1}+\cdot\cdot\cdot+dp_{n}\wedge dx_{n}.
\]
The phase space is the playground of classical mechanics it its Hamiltonian
formulation. The latter describes the motions of phase space points using
Hamilton's equation of motion%
\[
\frac{d}{dt}x_{j}(t)=\frac{\partial H}{\partial p_{j}}(x(t),p(t)\text{
\ },\text{ \ }\frac{d}{dt}p_{j}(t)=\frac{\partial H}{\partial x_{j}%
}(x(t),p(t)
\]
where $H$ is a function having suitable regularity properties. The phase space
flow $(f_{H}^{t})$ determined by these equation has a fundamental property: it
consists of symplectomorphisms, i.e. it preserves the symplectic form, that is
$(f_{H}^{t})^{\ast}\sigma=\sigma$. This implies that the Hamiltonian flow
preserves the symplectic capacities of subsets of $\mathbb{R}_{x}^{n}%
\times\mathbb{R}_{p}^{n}$, a property only known since 1985 following Gromov's
\cite{Gromov} so-called "non-squeezing theorem". The latter, which has
far-reaching consequences for the Hamiltonian dynamics, is a manifestation of
the "rigidity" of symplectomorphisms, and leads to a classical formulation of
the uncertainty principle. For instance, it implies that if we let the
Hamiltonian floe $(f_{H}^{t})$ act on a phase space ball $B^{2n}%
(z_{0},R):|z-z_{0}|\leq R)$, then the "shadow" )= orthogonal protection) of
$f_{H}^{t}(B^{2n}(z_{0},R))$ on any plane of conjugate coordinates
$x_{j},p_{j}$ will have are $\geq\pi R^{2}$ for all times $t_{.}$

In quantum mechanics, the notion of phase space point does not make sense
because of Heisenberg's uncertainty principle. Two substitutes for points are
commonly used in quantum mechanics. first, the "quantum blobs" we have
introduced in previous work \cite{blobs} which provide a coarse-graining of
classical phase space compatible with the uncertainty principle, secondly,
more commonly, the "squeezed coherent states" and their generalizations, which
are a functional representation of quantum blobs.

\subsubsection*{The Weyl--Wigner--Moyal representation of quantum mechanics}

This doesn't mean, however, that the classical phase space does not play any
role in quantum mechanics, on the contrary. Classical phase space plays a
fundamental role in what is called the "Weyl--Wigner--Moyal Interpretation of
quantum mechanics" and where a quasi-distribution, the Wigner transform (or
function) play a pivotal role in the so-called phase space quantum mechanics,
not only foe statistical purposes, but for a variety of theoretical problems,
also deeply influencing other areas like time-frequency analysis or the theory
of pseudo-differential operators (Eugene Wigner introduced his celebrated
quasi-distribution in \cite{Wigner} in a totally \textit{ad hoc} way,
acknowledging in a footnote inspiration from the physicist Leo Szilard,
however this sees to be a help to boost the career of the latter. we will
never know...). The so-called "Weyl quantization" has superadded other
quantization schemes (in particular the Born--Jordan--Heisenberg procedure)
mainly because of its relative simplicity, and the fact that its properties of
symplectic covariance (inherited from the Wigner transform, to which it is
closely related). This property reflects at the quantum level the canonical
invariance of Hamiltonian dynamics and justifies view that quantum mechanics
is a refinement of classical mechanics, as we explain now.

\subsubsection*{\textquotedblleft Quantum mechanics is a refinement of
classical mechanics\textquotedblright}

This statement, attributed to the mathematician George Mackey, goes straight
to the point even if it is not universally accepted by the physics community.
In fact, there is evidence for a "porosity" (M\'{e}linon \cite{Patrice})
between classical (Hamiltonian) mechanics and quantum theory. We have also
discussed these issues in a paper \cite{GoHi} with Hiley. in. One illustration
is the following, fact, well-known by ha harmonic analysis community, but less
so by most physicalists: let $H$ be a quadratic Hamiltonian function, e.g.
that of a generalized harmonic oscillator). The associated Hamilton equations
generate a liner flow, in fact a one-parameter subgroup $(S_{t})$ of the
symplectic group $\operatorname*{Sp}(n)$. Consider now metaplectic group
$\operatorname*{Mp}(n)$; it is a double covering of $\operatorname*{Sp}(n)$
consisting in unitary operators acting on square-integrable functions. A
general principle from the theory of covering spaces says that the
one-parameter group $(S_{t})$ is covered by a unique one-parameter subgroup
$(\widehat{S}_{t})$ of $\operatorname*{Mp}(n)$. It turns out that after some
calculations one finds that for every smooth function $\psi_{0}$ the
transformed function $\psi(x,t)=\widehat{S}\psi_{0}(x)$ satisfies the
equation
\[
i\hbar\frac{\partial\psi}{\partial t}=H(x,-i\hbar\nabla_{x})\hbar
\]
where $\hbar$ is a parameter which can be chosen arbitrarily. When $\hbar$
\ is chosen equal to $h/2\pi$ where $h$ is Planck's constant, then this
equation becomes Schr\"{o}dinger's equation, the fundamental equation of
quantum mechanics. But we observe that a mathematical equation does not
represent a physical theory unless its solution $\psi$ is given a physical
meaning. Still, the construction we just described is, in essence, what
Schr\"{o}dinger did, by using modern more sophisticated form. Schr\"{o}dinger
had, after Peter Debye's famous question "...but if there is a wave, what is
its wave equation?", the insight of manipulating the Hamilton--Jacobi equation
to extract his eponymous equation, well before the metaplectic group was born!

\section{\label{sec1}Quantum Blobs and Generalized Gaussians}

\subsection{What are quantum blobs?}

Quantum blobs are the most natural geometric substitutes for phase space
points. By definition, a quantum blob is the image by a linear symplectic
transformation of a phase space ball with radius $\sqrt{\hbar}$:%
\begin{equation}
Q(S,z_{0})=S(B^{2n}(z_{0},\sqrt{\hbar})). \label{blob1}%
\end{equation}
The quantum blob $Q(I,0)$ is the centered phase space ball $B^{2n}(\sqrt
{\hbar})$ and for every $S^{\prime}\in\operatorname*{Sp}(n)$ we have
$Q(S,z_{0})=Q(S^{\prime}S,z_{0})$. Let $\operatorname*{Blob}(n)$ be the set of
all quantum blobs in phase space $\mathbb{R}_{z}^{2n}$. The triple
$(\operatorname*{Blob}(n),\mathbb{R}_{z}^{2n},\operatorname*{proj})$ where
$\operatorname*{proj}$ is the projection $Q(S,z_{0})\longmapsto z_{0}$ is a
trivial fibration, which could be characterized as a "symplectic blow up" of
phase. leading to a coarse-graining of the latter. The interest of the notion
of quantum blob is (at least) twofold, as we explain in the following
subsections. The following factorization of symplectic matrices is essentially
a KAM decomposition. As we will see, it is very useful for describing quantum blobs:

\begin{lemma}
[Pre-Iwasawa]Let $S\in\operatorname*{Sp}(n)$. There exist unique symplectic
matrices
\begin{equation}
V_{P}=%
\begin{pmatrix}
I & 0\\
P & I
\end{pmatrix}
\text{ \ , \ }M_{L}=%
\begin{pmatrix}
L^{-1} & 0\\
0 & L
\end{pmatrix}
\label{vpml}%
\end{equation}
with $P,L\in\operatorname*{Sym}(n,\mathbb{R})$, $L>0$, and
\begin{equation}
R=%
\begin{pmatrix}
U & V\\
-V & V
\end{pmatrix}
\label{R}%
\end{equation}
a symplectic rotation such that
\begin{equation}
S=V_{P}M_{L}R. \label{iwa1}%
\end{equation}

\end{lemma}

\begin{proof}
For the justification of the symplectic matrices $V_{P}M_{L}R$ see next
section \ref{secsymp}. The proof computational; see \cite{iwa,dutta,Houde}.
Writing $S$ in block form $%
\begin{pmatrix}
A & B\\
C & D
\end{pmatrix}
$ and $R=$ these matrices are explicitly given by the formulas%
\begin{gather}
L=(AA^{T}+BB^{T})^{-1/2}\label{pl1}\\
P=-(CA^{T}+DB^{T})(AA^{T}+BB^{T})^{-1}\label{pl2}\\
U=(AA^{T}+BB^{T})^{-1/2}A\text{ \ },\text{ \ }V=(AA^{T}+BB^{T})^{-1/2}B.
\label{pl3}%
\end{gather}

\end{proof}

\begin{proposition}
Every quantum blob centered at $0$ is the image ball $B^{2n}\sqrt{\hbar})$ by
a product $V_{-P}M_{L}($ $P$ with $L\in\operatorname*{Sym}(n,\mathbb{R})$,
$L>0$. More generally,
\[
Q(S,z_{0})=T(z_{0})V_{-P}M_{L}(B^{2n}\sqrt{\hbar})
\]
where $T(z_{0}):z\longmapsto z+z_{0}$.
\end{proposition}

\begin{proof}
It is obvious in view of (\ref{iwa1}) since $R(B^{2n}\sqrt{\hbar}%
))=B^{2n}\sqrt{\hbar})=by$ rotational invariance.
\end{proof}

\subsection{The symplectic group $\operatorname*{Sp}(n)$ and its double cover
$\operatorname*{Mp}(n)$\label{secsymp}}

For a detailed study of the topics of this section see for instance
\cite{Folland,Birk,Leray}.

The symplectic group $\operatorname*{Sp}(n)$ consists of all linear
\ automorphisms $S$ of the symplectic space $(\mathbb{R}^{2n},\sigma$ which
preserve the symplectic form $\sigma$ that is $\sigma(Sz,Sz^{\prime}%
)=\sigma(z,z^{\prime})$ for all $z,z^{\prime}\in\mathbb{R}^{2n}$. Identifying
$S$ wit its matrix in the canonical basis of $\mathbb{R}^{2n}$ we have
$S\in\operatorname*{Sp}(n)$ if and only if $SJS^{T}=S^{T}JS=J$; \ it follows
that $\operatorname*{Sp}(n)$ is a closed subgroup of $GL(2n,\mathbb{R})$ and
hence a classical Lie group. Writing $S=%
\begin{pmatrix}
A & B\\
C & D
\end{pmatrix}
$, where the \textquotedblleft blocks\textquotedblright\ $A,B,C,D$ being
$n\times n$ matrices, we have $S\in\operatorname*{Sp}(n)$ if and only if
\begin{align}
A^{T}C,\text{ }B^{T}D\text{ \ \textit{are symmetric, and} }A^{T}D-C^{T}B  &
=I\label{cond12}\\
AB^{T},\text{ }CD^{T}\text{ \ \textit{are\ symmetric, and} }AD^{T}-BC^{T}  &
=I. \label{cond22}%
\end{align}
One shows that the group $\operatorname*{Sp}(n)$ is generated by the standard
symplectic matrix $J$ together with the matrices%
\[
V_{P}=%
\begin{pmatrix}
I_{n\times n} & 0\\
-P & I_{n\times n}%
\end{pmatrix}
\text{ \ },\text{ \ }M_{L}=%
\begin{pmatrix}
L^{-1} & 0\\
0 & L^{T}%
\end{pmatrix}
\]
where $P\in\operatorname*{Sym}(n,\mathbb{R})$ and $L\in GL(n,\mathbb{R})$.

A subgroup of $\operatorname*{Sp}(n)$ of particular interest is the group of
symplectic rotations
\begin{equation}
U(n)=\operatorname*{Sp}(n)\cap O(2m,\mathbb{R}). \label{USPO}%
\end{equation}
It is the image in $\operatorname*{Sp}(n)$ of the unitary group
$U(n,\mathbb{C})$ by the monomorphism
\[
\iota:u=+iB\longmapsto U=U=%
\begin{pmatrix}
A & B\\
-B & A
\end{pmatrix}
\]
(the conditions (\ref{cond12})--(\ref{cond22})) are satisfied since $u^{\ast
}u=uu^{\ast}u=I$). The elements of $U)n)$ are symplectic rotations:It also
follows from (\ref{cond12})--(\ref{cond22}) that the conditions that $U\in
U(n)$ if and only if they satisfy the equivalent conditions
\begin{align}
A^{T}B\text{ \textit{symmetric and }}A^{T}A+B^{T}B  &  =I\label{u1}\\
AB^{T}\text{ \textit{symmetric and }}AA^{T}+BB^{T}  &  =I. \label{u2}%
\end{align}

The symplectic group $\operatorname*{Sp}(n)$ is connected and contractible to
its maximal subgroup $U(n)$; the latter being isomorphic to the unitary group
$U(n,C)$ hence the group isomorphisms
\[
\pi_{1}(\operatorname*{Sp}(n))\simeq\pi_{1}(U(n,C))\simeq(\mathbb{Z},+).
\]
It follows that $\operatorname*{Sp}(n)$ has covering groups
$\operatorname*{Sp}_{q}(n)$ of all orders $q=2,3,...,+\infty$. It turns out
that the double cover $\operatorname*{Sp}_{2}(n)$ has a unitary representation
in $L^{2}(\mathbb{R}^{n})$ by the metaplectic group $\operatorname*{Mp}(n)$.
The covering mapping
\begin{equation}
\pi_{\operatorname*{Mp}}:\operatorname*{Mp}(n)\longrightarrow
\operatorname*{Sp}(n)\text{ \ },\text{ \ }\pi_{\operatorname*{Mp}}%
(\widehat{S})=S \label{pimp}%
\end{equation}
satisfies $\operatorname*{Ker}(\pi_{\operatorname*{Mp}})=\{-I,I\}$ and is
adjusted so that
\begin{equation}
\pi_{\operatorname*{Mp}}(\widehat{J})=J\text{ \ if \ }\widehat{J}%
\psi(x)=\left(  \tfrac{1}{2\pi i\hbar}\right)  ^{n/2}\int_{\mathbb{R}^{n}%
}e^{-\frac{i}{\hbar}x\cdot x^{\prime}}\psi(x^{\prime})dx^{\prime} \label{Jhat}%
\end{equation}
and one shows that $\operatorname*{Mp}(n)$ is generated by $\widehat{J}$
together with the unitary automorphisms
\begin{equation}
\widehat{V}_{P}\psi(x)=e^{-\frac{i}{2}Px\cdot x}\psi(x)\text{ \ ,
\ }\widehat{M}_{L,m}\psi(x)=i^{m}\sqrt{|\det L|}\psi(Lx) \label{vpmlhat}%
\end{equation}
where $P\in\operatorname*{Sym}(n,\mathbb{R})$ and $L\in GL(n,\mathbb{R})$; the
integer $m$ corresponds to a choice of $\arg\det L$. The operators
$\widehat{V}$ and $\widehat{M}$ cover the symplectic automorphisms $V_{P}$ and
$M_{L,m}$ defined above:
\[
\pi_{\operatorname*{Mp}}(\widehat{V}_{P})=V_{P}\text{ \ \textit{and} \ }%
\pi_{\operatorname*{Mp}}(\widehat{M}_{L,m})=M_{L}.
\]

\subsection{Quantum blobs and the strong uncertainty principle}

Consider a physical system (classical or quantum) with covariance matrix%
\begin{equation}
\Sigma=\int_{\mathbb{R}^{2n}}\underset{2n\times2n\text{ matrix}%
}{\underbrace{(z-\langle z\rangle)(z-\langle z\rangle)^{T}}}\rho(z)dz
\label{sigmaz1}%
\end{equation}
where $\langle z\rangle=\int_{\mathbb{R}^{2n}}z\rho(z)dz$ (it is assumed that
the second order moments $\int_{\mathbb{R}^{2n}}z_{\alpha}z_{\beta}\rho(z)dz$
exist for $1\leq\alpha,\beta\leq n$ ($\rho$ is assumed to be a
(quasi-)probability function. We will writing $\Sigma$ in block-matrix form
\begin{equation}
\Sigma=%
\begin{pmatrix}
\Delta(x,x) & \Delta(x,p)\\
\Delta(p,x) & \Delta(p,p)
\end{pmatrix}
\label{covell}%
\end{equation}
where $\Delta(x,x)=(\Delta(x_{j},x_{k}))_{1\leq j,k\leq n}$, \textit{etc}..
The well-known Robertson--Schr\"{o}dinger uncertainty principle%
\begin{equation}
(\Delta(x_{j},x_{j}))^{2}(\Delta(p_{j},p_{j}))^{2}\geq\Delta(x_{j},p_{j}%
)^{1}+\frac{1}{4}\hbar^{2} \label{RS1}%
\end{equation}
for $1\leq j\leq n$ is a consequence of the strong uncertainty principle
\cite{dutta,Birkbis,go09,golu09}.%
\begin{equation}
\Sigma+\frac{i\hbar}{2}J\text{ is positive semidefinite} \label{RS2}%
\end{equation}
which we abbreviate as $\Sigma+\frac{i\hbar}{2}J\geq0$. Both condition
(\ref{RS1}) and (\ref{RS2}) are trivially equivalent for $n=1$; when $n>1$
(\ref{RS2}) implies (\ref{RS1}); notice that (\ref{RS2}) is satisfied does not
imply that the system is "quantum": see our discussion in \cite{goluphysa} and
the references therein for counterexamples.

\begin{theorem}
\label{ThmBlob}The uncertainty principle (\ref{RS2}) is equivalent to the
following statement: The covariance ellipsoid%
\[
\Omega_{\Sigma}=z:\left\{  z\_\frac{1}{2}\Sigma^{-1}z\cdot z\leq1\right\}
\]
contains a quantum blob. When $\Omega_{\Sigma}$ is a quantum blob, then the
Robertson--Schr\"{o}dinger inequalities (\ref{RS1}) are saturated (i.e.
reduced to equalities).
\end{theorem}

\begin{proof}
See \cite{blobs,golu09}.
\end{proof}

This result has a deep topological meaning, using the notion of symplectic
capacity (whose definition is justified by Gromov's celebrated symplectic
non-squeezing theorem \cite{Gromov}; see our review in \cite{golu09}). An
\textit{intrinsic} symplectic capacity on $\mathbb{R}^{2n}$ assigns a
non-negative number (or $+\infty$) $c(\Omega)$ to every subset $\Omega
\subset\mathbb{R}^{2n}$; this assignment is subjected to the following properties:

\begin{itemize}
\item \textbf{Monotonicity:} If $\Omega\subset\Omega^{\prime}$ then
$c(\Omega)\leq c(\Omega^{\prime})$;

\item \textbf{Symplectic invariance:} If $f\in\operatorname*{Symp}%
(2n,\mathbb{R})$ then $c(f(\Omega))=c(\Omega)$;

\item \textbf{Conformality:} If $\lambda$ is a real number then $c(\lambda
\Omega)=\lambda^{2}c(\Omega)$;

\item \textbf{Normalization:} We have
\begin{equation}
c(B^{2n}(R))=\pi R^{2}=c(Z_{j}^{2n}(R)); \label{norm1}%
\end{equation}

\end{itemize}

Let $c$ be a symplectic capacity on the phase plane $\mathbb{R}^{2}$. We have
$c(\Omega)=\operatorname*{Area}(\Omega)$ when $\Omega$ is a connected and
simply connected surface. In the general case there exist infinitely many
intrinsic symplectic capacities, but they all agree on phase space ellipsoids
as we will see below.

\begin{corollary}
The uncertainty principle (\ref{RS2}) is satisfied if and only $c(\Omega
_{\Sigma})\geq\pi\hbar$ for 3every intrinsic symplectic capacity $c$ on
$\mathbb{R}^{2n}$.
\end{corollary}

\begin{proof}
It immediately follow3s from Theorem \ref{ThmBlob} using the normalization
priority (\ref{norm1}) of symplectic capacities.
\end{proof}

\subsection{Generalized Gaussians and quantum blobs}

In what follows denote $\phi_{0}^{\hbar}$ the standard coherent state
\begin{equation}
\phi_{0}^{\hbar}(x)=\psi_{I,0}(x)=(\pi\hbar)^{-n/4}e^{-|x|^{2}/2\hbar}.
\label{standard}%
\end{equation}
Let $X,Y\in\operatorname*{Sym}(n,\mathbb{R})$, $X$ positive definite: $X>0$.
We define the function
\begin{equation}
\psi_{XY}(x)=\left(  \tfrac{\det X}{(\pi\hbar)^{n}}\right)  ^{1/4}e^{-\frac
{1}{2\hbar}(X+iY)x\cdot x}. \label{squeezed}%
\end{equation}
These Gaussians can be obtained from using elementary metaplectic operators,
as follows from the obvious formula
\begin{equation}
\psi_{XY}=\widehat{S}_{XY}\phi_{0}^{\hbar}=\widehat{V}_{Y}\widehat{M}%
_{X^{1/2}}\phi_{0}^{\hbar} \label{fixy}%
\end{equation}
where $\widehat{V}_{Y}$ and $\widehat{M}_{X^{1/2}}=\widehat{M}_{X^{1/2},0}$
are defined by (\ref{vpmlhat}).Note that neither the operators $\widehat{S}%
_{XY}n$ nor their projections $S_{XY}$ \ form a group if $X$ and $X^{\prime}$
are symmetric.

More generally, we define the displaced Gaussians
\begin{equation}
\psi_{XY}=\widehat{T}(z_{0})\psi_{XY}\text{ \ \ },\text{ \ \ }z_{0}%
=(x_{0},p_{0}) \label{txy}%
\end{equation}
to which the transformations above are easily extended.

\begin{definition}
We denote by $\operatorname*{Gauss}(n)$ the set of all functions\ $\psi
_{XY}^{z_{0}}$ (\textit{i.e.} the collection of all functions $c\psi
_{XY}^{z_{0}}$where $c\in\mathbb{C}$ with \ $|c|=1$). The subset of
$\operatorname*{Gauss}(n)$ consisting of all $|\psi_{XY}\rangle$ is denoted by
$\operatorname*{Gauss}_{0}(n)$.
\end{definition}

The following result identifies the sets $\operatorname*{Blob}(n$ and
$\operatorname*{Gauss}(n)$:

\begin{theorem}
\label{Prop2}The mapping%
\begin{equation}
\Gamma:\operatorname*{Blob}(n)\longrightarrow\operatorname*{Gauss}(n)
\label{gamma}%
\end{equation}
defined by
\[
\Gamma:Q(z_{0},S_{XY}))\longmapsto\psi_{XY}^{z_{0}}=\widehat{T}(z_{0}%
)\widehat{S}_{XY}\phi_{0}^{\hbar}%
\]
where $S_{XY}=V_{-Y}M_{X^{-1/2}}$ and $\widehat{S}_{XY}=\widehat{V}%
_{Y}\widehat{M}_{X^{-1/2}}$ is a bijection.
\end{theorem}

\begin{proof}
Using the intertwining formulas
\[
\widehat{S}\widehat{T}(z)=\widehat{T}(Sz)\widehat{S}\text{ \ },\text{
\ }ST(z)=T(Sz)S
\]
it suffices to consider the case $z=0$. Let us to show that the restriction
\begin{gather*}
\Gamma_{0}:\operatorname*{Blob}\nolimits_{0}(n)\longrightarrow
\operatorname*{Gauss}\nolimits_{0}(n)\\
S_{XY}B^{2n}(\sqrt{\hbar})\longmapsto\widehat{S}_{XY}\phi_{0}^{\hbar}\rangle
\end{gather*}
is a bijection. Firstly, $\Gamma_{0}$ is a well-defined mapping since every
quantum blob $Q_{S}$ can be written, using the pre-Iwasawa factorization as
(\ref{iwa1}), as%
\[
Q(0,S_{XY})=S_{XY}(B^{2n}(\sqrt{\hbar})=V_{Y}M_{X^{-1/2}}(B^{2n}(\sqrt{\hbar
}).
\]
Similarly every Gaussian function $\psi_{XY}$ can be written as $\psi
_{XY}=\widehat{S}_{XY}\phi_{0}^{\hbar}$, showing at the same time that
$\Gamma_{0}$ is surjective. To show that $\Gamma_{0}$ is \ bijection there
remains to \ prove injectivity, that is if $\widehat{S}_{XY}\phi_{0}^{\hbar
}=\widehat{S}_{X^{\prime},Y^{\prime}}^{\prime}\phi_{0}^{\hbar}$ then
$S_{XY}B^{2n}(\sqrt{\hbar}))=S_{X^{\prime},Y^{\prime}}B^{2n}(\sqrt{\hbar}))$.
In view of the rotational symmetry of the standard coherent state $\phi
_{0}^{\hbar}$ we must have $\widehat{S}_{XY}=\widehat{S}_{X^{\prime}%
,Y^{\prime}}\widehat{R}$ where $\widehat{R}\in\operatorname*{Mp}(n)$ \ covers
a symplectic rotation $R\in\operatorname*{Sp}(n)\cap O(2n,\mathbb{R})$, hence
$S_{X^{\prime},Y^{\prime}}=S_{XY}R$ and the injectivity follows since
$R(B^{2n}(\sqrt{\hbar})))=B^{2n}(\sqrt{\hbar}))$.
\end{proof}

Notice that the bijection $\Gamma$ satisfies
\[
\Gamma(S^{\prime}Q(S,z_{0})=\Gamma(|\widehat{S^{\prime}}\psi_{XY}^{z_{0}%
}\rangle
\]
for all $S,S^{\prime}\in\operatorname*{Sp}(n)$ if $\widehat{S^{\prime}}$
covers $S^{\prime}$. Similarly every Gaussian function $\psi_{XY}$ can be
written as $\psi_{XY}=\widehat{S}_{XY}\phi_{0}^{\hbar}$, showing at the same
time that $\Gamma_{0}$ is surjective. To show that $\Gamma_{0}$ is \ bijection
there remains to \ prove injectivity, that is if $\widehat{S}_{XY}\phi
_{0}^{\hbar}=\widehat{S}_{X^{\prime},Y^{\prime}}^{\prime}\phi_{0}^{\hbar}$
then $S_{XY}B^{2n}(\sqrt{\hbar}))=S_{X^{\prime},Y^{\prime}}B^{2n}(\sqrt{\hbar
}))$. In view of the rotational symmetry of the standard coherent state
$\phi_{0}^{\hbar}$ we must have $\widehat{S}_{XY}=\widehat{S}_{X^{\prime
},Y^{\prime}}\widehat{R}$ where $\widehat{R}\in\operatorname*{Mp}(n)$ \ covers
a symplectic rotation $R\in\operatorname*{Sp}(n)\cap O(2n,\mathbb{R})$, hence
$S_{X^{\prime},Y^{\prime}}=S_{XY}R$ and the injectivity follows since
$R(B^{2n}(\sqrt{\hbar})))=B^{2n}(\sqrt{\hbar}))$.

\subsection{The canonical group of a quantum blob}

The standard Gaussian $\phi_{0}^{\hbar}$ satisfies the partial differential
equation%
\begin{equation}
\widehat{H}_{0}\phi_{0}^{\hbar}=\frac{1}{2}(-\hbar^{2}\nabla_{x}^{2}%
+|x|^{2})\phi_{0}^{\hbar}=\frac{1}{2}n\hbar\phi_{0}^{\hbar} \label{Ho}%
\end{equation}
(it is the stationary Schr\"{o}dinger equation for the ground state of the
isotropic $n$-dimensional harmonic oscillator with mass one). The solutions of
the corresponding Schr\"{o}dinger equation%
\[
i\hbar\frac{\partial\psi}{\partial t}(x,t)=\widehat{H}_{0}\psi(x,t)
\]
are given by $\psi(x,t)=\widehat{S}_{t}\psi(x,0)$ where $(\widehat{S})$ the
\ one-parameter group of metaplectic operators covering $(S)$.

We now extend the previous result to the case of general Gaussians (we are
following here our constructions in \cite{JGP}). Recall that we have defined
the Gaussian function
\begin{equation}
\psi_{XY}(x)=\left(  \tfrac{1}{\pi\hbar}\right)  ^{n/4}(\det X)^{1/4}%
e^{-\tfrac{1}{2\hbar}(X+iY)x\cdot x} \label{psixy}%
\end{equation}
where $X,Y$ are symmetric and $X>0.$ It is the most general (up to
translations) function whose Wigner transforms are positive. This function is
normalized to unity: $||\psi_{XY}^{\gamma}||_{L^{2}}=1$. We remark that
$\psi_{XY}$ is a solution of the eigenvalue liquation
\begin{equation}
\widehat{H}_{XY}\psi_{XY}=\frac{1}{2}\hbar\operatorname*{Tr}(X)\psi_{XY}
\label{H2}%
\end{equation}
where $\widehat{H}_{XY}$ \ is the second order partial differential operator%
\begin{equation}
\widehat{H}_{X%
\acute{}%
Y}=\frac{1}{2}(-i\hbar\nabla_{x}+Yx)^{2}+X^{2}x\cdot x \label{hhat}%
\end{equation}
which is the Weyl quantization of the quadratic polynomial
\begin{equation}
H_{XY}(x,p)=\frac{1}{2}\left(  (p+Yx)^{2}+X^{2}x\cdot x\right)  . \label{xyh}%
\end{equation}
The latter can be rewritten in matrix form as
\begin{equation}
H_{XY}(z)=\frac{1}{2}M_{XY}z\cdot z \label{gf4}%
\end{equation}
where $M_{XY}$ is the symmetric positive-definite matrix%
\begin{equation}
M_{XY}=%
\begin{pmatrix}
X^{2}+Y^{2} & Y\\
Y & I
\end{pmatrix}
.=(S_{XY}^{-1})^{T}D_{X}S_{XY}^{-1} \label{mxy}%
\end{equation}
where
\begin{equation}
S_{XY}=V_{Y}M_{X^{1/2}}\text{ \ },,\text{\ \ }D_{X}=%
\begin{pmatrix}
X & 0\\
0 & X
\end{pmatrix}
. \label{mfs}%
\end{equation}
Note that
\[
H_{XY}\circ S_{XY}(z)=\frac{1}{2}D_{X}z\cdot z=\frac{1}{2}Xx\cdot x+\frac
{1}{2}Xp\cdot p.
\]
The flow generated by this the Hamiltonian function $H_{XY}$ is a
one-parameter subgroup $(S_{t}^{XY})$ of $\operatorname*{Sp}(n)$ which we call
the \emph{canonical group} of the quantum blob represented by $\psi_{XY}$,

\begin{theorem}
Let $(\widehat{S}_{t}^{XY}))$ be the one-parameter subgroup of
$\operatorname*{Mp}(n)$ covering the canonical group $(S_{t}^{XY})$. It has
the following properties. We have
\begin{equation}
\widehat{S}_{t}^{XY}\psi_{XY}=\exp\left[  \frac{it}{2\hbar}\hbar
\operatorname*{Tr}(X)\right]  \psi_{XY}. \label{stxy}%
\end{equation}
(ii) If
\end{theorem}

\begin{proof}
The operator $\widehat{H}_{XY}$ is the Weyl quantization of $H_{XY}$; then
$\widehat{S}_{t}^{XY}\psi_{XY}$ is the solution of Schr\"{o}dinger's equation%
\begin{equation}
i\hbar\frac{\partial}{\partial t}\widehat{S}_{t}^{XY}\psi_{XY}=\widehat{H}%
_{XY}(\widehat{S}_{t}^{XY}\psi_{XY}) \label{erwin}%
\end{equation}
with initial condition $\widehat{S}_{0}^{XY}\psi_{XY}=\psi_{XY}$, By the
symplectic covariance of Weyl operators we have
\[
\widehat{S}_{-t}^{XY}\widehat{H}_{XY}\widehat{S}_{t}^{XY}=\widehat{H_{XY}\circ
S_{t}^{XY}}=\widehat{H}_{XY}%
\]
(the second equity because the Hamiltonian is constant along the
trajectories); it follows that Schr\"{o}dinger's equation (\ref{erwin})
becomes
\begin{equation}
i\hbar\frac{\partial}{\partial t}\widehat{S}_{t}^{XY}\psi_{XY}=\widehat{S}%
_{t}^{XY}\widehat{H}_{XY}(\psi_{XY}) \label{erwin2}%
\end{equation}
(ii)recalling (formula (\ref{H2})) that $\widehat{H}_{XY}\psi_{XY}=\frac{1}%
{2}\hbar\operatorname*{Tr}(X)\psi_{XY}$ this reduces to%
\begin{equation}
i\hbar\frac{\partial}{\partial t}\widehat{S}_{t}^{XY}\psi_{XY}=\frac{1}%
{2}\hbar\operatorname*{Tr}(X)\widehat{S}_{t}^{XY}\psi_{XY}%
\end{equation}
which has solution%
\[
\widehat{S}_{t}^{XY}\psi_{XY}=\exp\left[  \frac{it}{2\hbar}\hbar
\operatorname*{Tr}(X)\right]  .
\]
(ii)
\end{proof}

\section{\label{sec2}We--Wigner--Moyal Theory}

\subsection{Weyl quantization}

Let $\widehat{A}$ be a linear continuous operator $\mathcal{S}(\mathbb{R}%
^{n})\longrightarrow\mathcal{S}^{\prime}(\mathbb{R}^{n})$. It follows from
Schwartz's kernel theorem that there exists a distribution $K\in
\mathcal{S}^{\prime}(\mathbb{R}^{n})\times\mathcal{S}^{\prime}(\mathbb{R}%
^{n})$ \ such that $\langle\widehat{A}\psi,\phi\rangle=\langle K,\psi
\otimes,\phi\rangle$.

\begin{definition}
The Weyl symbol of the operator $\widehat{A}$ is he distribution
$a\in\mathcal{S}^{\prime}(\mathbb{R}^{2n})$ defined the integral
\begin{equation}
a(x,p)=\int_{\mathbb{R}^{n}}e^{-\frac{i}{\hbar}p\cdot y}K(x+\tfrac{1}%
{2}y,x-\tfrac{1}{2}y)dy \label{AK9}%
\end{equation}
view as $/2\pi\hbar)^{n}F_{y\longrightarrow p}K(x+\tfrac{1}{2}\cdot
,x-\tfrac{1}{2}\cdot)$. We will write $\widehat{A}=\operatorname*{Op}%
_{\mathrm{Weyl}}(a)$ and call the "Weyl operator with symbol $a$".
\end{definition}

Using the Fourier inversion formula we have, conversely,
\begin{equation}
K(x,y)=\left(  \tfrac{1}{2\pi\hbar}\right)  ^{n}\int_{\mathbb{R}^{n}}%
e^{\frac{i}{\hbar}p\cdot(x-y)}a(\tfrac{1}{2}(x+y),p)dp \label{KA9}%
\end{equation}
from which one gets the formal textbook integral .representation of
$\operatorname*{Op}_{\mathrm{Weyl}}(a)$:
\begin{equation}
\widehat{A}\psi(x)=\left(  \tfrac{1}{2\pi\hbar}\right)  ^{n}\iint%
\nolimits_{\mathbb{R}^{n}\times\mathbb{R}^{n}}e^{\frac{i}{\hbar}p\cdot
(x-y)}a(\tfrac{1}{2}(x+y),p)\psi(y)dydp\text{.} \label{weylahat}%
\end{equation}

If $a=1$ then $\widehat{A}=\operatorname*{Op}_{\mathrm{Weyl}}(1)$ is the
identity operator: this follows that the fact that the kernel is in this case
$K(x,y)=\delta(x-y)$.

Weyl operators have two useful harmonic representations using the Heisenberg
displacement and Grossmann-Royer reflection operators $\widehat{T}(z_{0}%
)\psi=e^{-\frac{i}{\hslash}\sigma(\hat{z},z_{0})}$ and $\widehat{T}%
_{\text{GR}}(z_{0})$ which are explicitly given by
\begin{gather}
\widehat{T}(z_{0})\psi(x)=e^{\tfrac{i}{\hslash}(p_{0}\cdot x-\tfrac{1}{2}%
p_{0}\cdot x_{0})}\psi(x-x_{0})\label{HW}\\
\widehat{T}_{\text{GR}}(z_{0})\psi(x)=e^{\frac{2i}{\hbar}p_{0}\cdot(x-x_{0}%
)}\psi(2x_{0}-x). \label{GR}%
\end{gather}
Observe that the Heisenberg operator satisfies the well-known relations%
\begin{align}
\widehat{T}(z_{0})\widehat{T}(z_{1})  &  =e^{\tfrac{i}{\hslash}\sigma
(z_{0},z_{1})}\widehat{T}(z_{1})\widehat{T}(z_{0})\label{HW1}\\
\widehat{T}(z_{0}+z_{1})  &  =e^{-\tfrac{i}{2\hslash}\sigma(z_{0},z_{1}%
)}\widehat{T}(z_{0})\widehat{T}(z_{1}). \label{HW2}%
\end{align}

\begin{theorem}
Let $a\in\mathcal{S}^{\prime}(\mathbb{R}^{2n})$. We have
\begin{align}
\operatorname*{Op}\nolimits_{\mathrm{Weyl}}(a)  &  =\left(  \tfrac{1}%
{2\pi\hbar}\right)  ^{n}\int_{\mathbb{R}^{2n}}a_{\sigma}(z_{0})\widehat{T}%
(z_{0})dz_{0}\label{opHW}\\
\operatorname*{Op}\nolimits_{\mathrm{Weyl}}(a)  &  =\left(  \tfrac{1}{\pi
\hbar}\right)  ^{n}\int_{\mathbb{R}^{2n}}a(z_{0})\widehat{T}_{\text{GR}}%
(z_{0})dz_{0} \label{opGR}%
\end{align}
where the integrals are interpreted in the sense of Bochner (that is, operator
valued integrals). The distribution $a_{\sigma}$ is the symplectic Fourier
transform of the Weyl symbol $a$:
\begin{equation}
a_{\sigma}(z_{0})=F_{\sigma}a(z)=\left(  \tfrac{1}{2\pi\hbar}\right)
\int_{\mathbb{R}^{2n}}e^{-\frac{i}{\hbar}\sigma(z,z^{\prime})}a(z^{\prime
})dz^{\prime}. \label{SFT}%
\end{equation}

\end{theorem}

A characteristic property of Weyl operators is their symplectic covariance:

\begin{theorem}
Let $a\in\mathcal{S}^{\prime}(\mathbb{R}^{2n})$ and $S\in\operatorname*{Sp}%
(n)$. (i) We have
\[
\operatorname*{Op}\nolimits_{\mathrm{Weyl}}(a\circ S^{-1})=\widehat{S}%
\operatorname*{Op}\nolimits_{\mathrm{Weyl}}(a)\widehat{S}^{-1}%
\]
where $\widehat{S}\in Mp(n)$ is any of the two metaplectic operators covering
$S$. (ii) Weyl operators are the only pseudodifferential operators enjoying
this symplectic covariance property.
\end{theorem}

\begin{proof}
(i) It follows from (\ref{opHW}) using the intertwining relation%
\begin{equation}
\widehat{S}\widehat{T}(z_{0})\widehat{S}^{-1}=\widehat{T}(Sz_{0}) \label{4418}%
\end{equation}
valid for every $z_{0}\in\mathbb{R}^{2n}$. (ii) See \cite{Stein,Wong}.
\end{proof}

When defined, the compose of two Weyl operators is itself a Weyl operator:The
product $\widehat{A}=\operatorname*{Op}\nolimits_{\mathrm{Weyl}}(a$ and
$\widehat{B}=\operatorname*{Op}\nolimits_{\mathrm{Weyl}}(b)$, the
$\widehat{C}=\widehat{A}\widehat{B}$ has Weyl symbol%
\begin{equation}
c(z)=\left(  \tfrac{1}{4\pi\hbar}\right)  ^{2n}\iint\nolimits_{\mathbb{R}%
^{4n}}e^{\frac{i}{2\hbar}\sigma(z^{\prime},z^{\prime\prime})}a(z+\tfrac{1}%
{2}z^{\prime})b(z-\tfrac{1}{2}z^{\prime\prime})dz^{\prime}dz^{\prime\prime}
\label{compo8}%
\end{equation}
and its symplectic Fourier transform of $c$ is given by%
\begin{equation}
c_{\sigma}(z)=\left(  \tfrac{1}{2\pi\hbar}\right)  ^{n}\int_{\mathbb{R}^{2n}%
}e^{\frac{i}{2\hbar}\sigma(z,z^{\prime})}a_{\sigma}(z-z^{\prime})b_{\sigma
}(z^{\prime})dz^{\prime}. \label{cecomp}%
\end{equation}
One often writes $c=a\star_{\hbar}b$ and calls it the Moyal product; in the
form can be rewritten
\begin{equation}
a\star_{\hbar}b(z)=\left(  \tfrac{1}{\pi\hbar}\right)  ^{2n}\iint%
\nolimits_{\mathbb{R}^{4n}}e^{\frac{2i}{\hbar}\partial\sigma(u,z,v)}%
a(u)b(v)dudv \label{Moyalprod}%
\end{equation}
where $\partial\sigma=\sigma(u,z)-\sigma(u,v)+\sigma(z,v).is$ an antisymmetric cocycle.

\subsection{The Wigner and ambiguity transforms}

We begin by defining a polarized version of the Wigner transform:

\begin{definition}
Let $(\psi,\phi)\in L^{2}(\mathbb{R}^{n})\times L^{2}(\mathbb{R}^{n})$. The
cross-Wigner transform of $(\psi,\phi)\in L^{1}(\mathbb{R}^{n})$ is or defined
by%
\begin{equation}
W(\psi,\phi)(z)=\left(  \tfrac{1}{\pi\hbar}\right)  ^{n}(\widehat{T}%
_{\text{GR}}(z)\psi|\phi)_{L^{2}}. \label{wigroyer}%
\end{equation}
The function $W\psi=W(\psi,\psi)$ is called the \textit{Wigner} transform\ (or
function) of $\psi$.
\end{definition}

A straightforward calculation using (\ref{GR}) leads to the traditional
formulas%
\begin{gather}
W(\psi,\phi)(z)=\left(  \tfrac{1}{2\pi\hbar}\right)  ^{n}\int_{\mathbb{R}^{n}%
}e^{-\frac{i}{\hbar}p\cdot y}\psi(x+\tfrac{1}{2}y)\overline{\phi(x-\tfrac
{1}{2}y)}dy\label{crosswig2}\\
W\psi(z)=\left(  \tfrac{1}{2\pi\hbar}\right)  ^{n}\int_{\mathbb{R}^{n}%
}e^{-\tfrac{i}{\hbar}p\cdot y}\psi(x+\tfrac{1}{2}y)\overline{\psi(x-\tfrac
{1}{2}y)}dy. \label{wig2}%
\end{gather}
Notice that $W(\psi,\phi)=\overline{W(\phi,\psi)}$ while $W\psi$ \textit{is
always a real function}. The Wigner transform determines that function up to
complex factor with modulus one:%
\[
W\psi=W\psi^{\prime}\Longleftrightarrow\psi=e^{i\varphi}\psi^{\prime}%
,\varphi\in\mathbb{R}.
\]

The Wigner transform satisfies the following covariance properties:%
\begin{gather}
W(\widehat{T}(z_{0})\psi),\widehat{T}(z_{0})\phi))(z)=W(\psi,\phi
))(z-z_{0})\label{trans1}\\
W(\widehat{S}\psi,\widehat{S}\phi)(z)=W(\psi,\phi)(S^{-1}z) \label{wigsym}%
\end{gather}
for $,\widehat{S}\in\operatorname*{Mp}(n)$ with projection $S\in
\operatorname*{Sp}(n)$.

The interest of the Wigner transform \ in quantum mechanics comes from the
fact that it can be viewed as a probability distribution. For instance, it
satisfies the \textquotedblleft marginal properties\textquotedblright: if
$\psi\in L1(\mathbb{R}^{n})\cap L^{2}(\mathbb{R}^{n})$ then%
\begin{equation}
\int_{\mathbb{R}^{n}}W\psi(x,p)dp=|\psi(x)|^{2}\text{ \ },\text{ \ }%
\int_{\mathbb{R}^{n}}W\psi(x,p)dx=|F\psi(p)|^{2} \label{marginal}%
\end{equation}
and hence
\begin{equation}
||W\psi||_{L^{1}(\mathbb{R}^{2n})}=||\psi||_{L^{2}.} \label{l1l2}%
\end{equation}
In addition, we have the important property that relates Weyl calculus and
Wigner transforms

\begin{theorem}
Let $(\psi,\phi)\in\mathcal{S}(\mathbb{R}^{n})\times\mathcal{S}(\mathbb{R}%
^{n})$ and assume that $\operatorname*{Op}_{\mathrm{Weyl}}(a)$ is a mapping
$\mathcal{S}(\mathbb{R}^{n})\longrightarrow L^{2}(\mathbb{R}^{n})$. We have
\begin{equation}
(\widehat{A}\psi|\phi)_{L^{2}(\mathbb{R}^{n})}=\int_{\mathbb{R}^{2n}%
}a(z)W(\psi,\phi)(z)dz \label{formula2}%
\end{equation}
and hence
\begin{equation}
(\widehat{A}\psi|\psi)_{L^{2}(\mathbb{R}^{n})}=\int_{\mathbb{R}^{2n}}%
a(z)W\psi((z)dz.
\end{equation}

\end{theorem}

The cross-Wigner transform satisfies the\textit{ Moyal identity}%
\begin{equation}
(W(\psi,\phi)|W(\psi^{\prime},\phi^{\prime}))_{L^{2}(\mathbb{R}^{2n})}=\left(
\tfrac{1}{2\pi\hbar}\right)  ^{n}(\psi|\psi^{\prime})_{L^{2}}\overline
{(\phi|\phi^{\prime})_{L^{2}}} \label{wigwig}%
\end{equation}
which yields, in particular,%
\begin{equation}
||W\psi||_{L^{2}(\mathbb{R}^{2n})}=\left(  \tfrac{1}{2\pi\hbar}\right)
^{n/2}||\psi||_{L^{2}(\mathbb{R}^{n})}\text{.} \label{wigwigwig}%
\end{equation}

Another, related, transform is the cross-ambiguity function:

\begin{definition}
Let $(\psi,\phi)\in\mathcal{S}(\mathbb{R}^{n})$. The function
$\operatorname*{Amb}(\psi,\phi)$ defined by
\begin{equation}
\operatorname*{Amb}(\psi,\phi)(-z)=\left(  \tfrac{1}{2\pi\hbar}\right)
^{n}(\widehat{T}(z)\psi|\phi)_{L^{2}} \label{ambigter}%
\end{equation}
is called the cross-ambiguity function; the function $\operatorname*{Amb}%
\psi,=\operatorname*{Amb}(\psi,\psi)$ is called the (radar) ambiguity function.
\end{definition}

The cross-Wigner and cross-ambiguity functions are related in two ways. First,
they are symplectic Fourier transforms of each other
\begin{equation}
W(\psi,\phi)=F_{\sigma}\operatorname*{Amb}(\psi,\phi\text{ \ },\text{
\ }W(\psi,\phi)=F_{\sigma}\operatorname*{Amb}(\psi,\phi)) \label{wsft}%
\end{equation}
(the symplectic Fourier transform (\ref{SFT}) \ is an involution
$\mathcal{S}^{\prime}(\mathbb{R}^{n})\longrightarrow\mathcal{S}^{\prime
}(\mathbb{R}^{n})$); secondly we have the functional relation
\begin{equation}
\operatorname*{Amb}(\psi,\phi)(z)=2^{-n}W(\psi,\phi^{\vee})(\tfrac{1}{2}z)
\label{fouwig2}%
\end{equation}
where $\phi^{\vee}(x)=\phi(-x)$. In particular, if $\psi$ is an even function%
\begin{equation}
\psi=\psi^{\vee}\Longrightarrow\operatorname*{Amb}\psi(z)=2^{-n}W\psi
(\tfrac{1}{2}z). \label{even}%
\end{equation}
It follows from (\ref{wsft}), taking into account the unitarity of the
symplectic Fourier transform, that the ambiguity function also satisfies a
Moyal identity:%
\begin{equation}
(\operatorname*{Amb}(\psi,\phi)|\operatorname*{Amb}(\psi^{\prime},\phi
^{\prime}))_{L^{2}(\mathbb{R}^{2n})}=\left(  \tfrac{1}{2\pi\hbar}\right)
^{n}(\psi|\psi^{\prime})_{L^{2}}\overline{(\phi|\phi^{\prime})_{L^{2}}}.
\label{ambamb}%
\end{equation}

We have $W\psi_{XY}\in\mathcal{S}(\mathbb{R}^{2n})$ hence $a\ast W\psi_{XY}%
\in\mathcal{S}^{\prime}(\mathbb{R}^{2n})$ so that $\operatorname*{Op}%
\nolimits_{\mathrm{XY}}(a)$ is well defined as a Weyl operator. In fact, the
Wigner transform of $\psi_{XY}$ is given by \cite{Birk,WIGNER,Littlejohn}
\begin{equation}
W\psi_{XY}(z)=\left(  \tfrac{1}{\pi\hbar}\right)  ^{n}e^{-\tfrac{1}{\hbar
}Gz\cdot z} \label{wxy}%
\end{equation}
where $G$ is the symmetric positive definite and symplectic matrix
\begin{equation}
G=%
\begin{pmatrix}
X+YX^{-1}Y & YX^{-1}\\
X^{-1}Y & X^{-1}%
\end{pmatrix}
=S^{T}S\
\end{equation}
where \
\begin{equation}
S=%
\begin{pmatrix}
X^{1/2} & 0\\
X^{-1/2}Y & X^{-1/2}%
\end{pmatrix}
\in\operatorname*{Sp}(n). \label{sts}%
\end{equation}

We make two preliminary remarks: when $\psi_{XY}$ is the standard Gaussian
$\phi_{0}$ (i.e. $\phi_{0}(x)=(\pi\hbar)^{-n/4}e^{-x^{2}/2\hbar}$) then its
Wigner transform is \ $W\phi_{0}(z)=(\pi\hbar)^{-n}e^{-|z|^{2}/\hbar}$ so
that
\[
a\ast W\phi_{0}(z)=\int_{\mathbb{R}^{2n}}a(z)(\pi\hbar)^{-n}e^{-|z-z_{0}%
|^{2}/\hbar}dz_{0}%
\]
which is the usual anti-Wick symbol (as defined in Shubin \cite{sh87}). We
will come back to this topic in the forthcoming sections.

A well-known related topic is that of the Husimi function. Let $\psi,\phi\in
L^{2}(\mathbb{R}^{n})$; by definition the Husimi function of the pair
$(\psi,\phi)$ is defined by%
\begin{equation}
W_{\mathrm{Hus}}(\psi,\phi)=W(\psi,\phi)\ast W\phi_{0} \label{huswig}%
\end{equation}
where $\phi$ is the standard coherent state. Its main interest comes from the
fact that it is a positive function:%
\begin{equation}
W_{\mathrm{Hus}}(\psi,\phi)\geq0\text{ \ for \ }\psi,\phi\in L^{2}%
(\mathbb{R}^{n}) \label{huspos1}%
\end{equation}
following the well-known result: \cite{bruijn}: for $\psi,\phi\in
L^{2}(\mathbb{R}^{n})$. We have%
\begin{equation}
W\psi\ast W\phi=|A(\widetilde{\psi},\phi)|^{2}=|F_{\sigma}W(\widetilde{\psi
},\phi)|^{2}. \label{wicon}%
\end{equation}

\subsection{A god function space: the Feichtinger algebra.}

The Wigner formalism allows to define an algebra of functions on configuration
space well adapted for the study of phase space quantum mechanics. This
algebra -- the Feichtinger algebra --- is usually defined in terms of the
so-called short-time Fourier transform (STFT) \cite{Gro}, but we will rather
use the Wigner transform, to which the STFT is closed related (see our
presentation in \cite{Birkbis}).

The Feichtinger algebra, of which we give here a simple (non-traditional)
definition is a particular case of thew more general notion of Feichtinger's
modulation spaces \cite{Hans1,Hans2,Hans3,Hans4,Hans5}, \cite{Gro}. These
spaces play an important role in time-frequency analysis, but are yet
underestimated in quantum mechanics.

\begin{definition}
The Feichtinger algebra $S_{0}(\mathbb{R}^{n}$ consists of all function
$\psi\in L^{2}(\mathbb{R}^{n})$ such that $W\psi\in L^{1}(\mathbb{R}^{2n})$.
\end{definition}

Therere seems t o be rub with this definition, because it is not clear why it
should define even a vector space (the Wigner transform is not linear!).
However \cite{Gro}, (\cite{Birkbis}, Ch.16):

\begin{proposition}
(i) We have $\psi\in S_{0}(\mathbb{R}^{n})$ if and only if there exists one
window $\phi$ such that $W(\psi,\phi)\in L^{1}(\mathbb{R}^{2n})$, in which
case we have $W(\psi,\phi)\in L^{1}(\mathbb{R}^{2n})$ for all windows
$\phi\phi\in\mathcal{S}(\mathbb{R}^{n})$; (ii) If $W(\psi,\phi)\in
L^{1}(\mathbb{R}^{2n})$ then both $\psi$ and $\phi$ are in $S_{0}%
(\mathbb{R}^{n})$; (iii) The functions $\psi\longmapsto||\psi||_{\phi,S_{0}}$
($\phi\in\mathcal{S}(\mathbb{R}^{n})$) defined by
\[
||\psi||_{\phi,S_{0}}=||W(\psi,\phi)||_{L^{1}(\mathbb{R}^{2n})}%
\]
are equivalent norms on $S_{0}(\mathbb{R}^{n}),$ which is a Banach space for
the apology thus defined. (iv) $S_{0}(\mathbb{R}^{n})$ is an algebra for both
usual (pointwise) multiplication and convolution.
\end{proposition}

We have the inclusions%
\begin{equation}
\mathcal{S}(\mathbb{R}^{n})\subset S_{0}(\mathbb{R}^{n})\subset C^{0}%
(\mathbb{R}^{n})\cap L^{1}(\mathbb{R}^{n})\cap L^{2}(\mathbb{R}^{n}).
\label{inclo}%
\end{equation}
The Schwartz space $\mathcal{S}(\mathbb{R}^{n})$ is dense in $S_{0}%
(\mathbb{R}^{n})$.The Feichtinger algebra $S_{0}(\mathbb{R}^{n})$ contains
continuous non-differentiable functions, for instance%
\[
\psi(x)=\left\{
\begin{array}
[c]{c}%
1-|x|\text{ \textit{if} }|x|\leq1\\
0\text{ \textit{if} }|x|>1
\end{array}
\right.  .
\]

\begin{proposition}
\label{169}Let $\psi\in S_{0}(\mathbb{R}^{n})$. We have (i) $\widehat{S}%
\psi\in S_{0}(\mathbb{R}^{n})$ for every $\widehat{S}\in\operatorname*{Mp}%
(n)$,; (ii) $\widehat{T}(z_{0})\psi\in S_{0}(\mathbb{R}^{n})$ for every
$z_{0}\in\mathbb{R}^{2n}$ . (iii) We have $\lim_{|x|\rightarrow\infty}\psi=0$
hence $\psi$ is bounded.
\end{proposition}

\begin{proof}
(Cf. \cite{Birkbis}, Ch. 16). (i) We have $\psi\in S_{0}(\mathbb{R}^{n})$ if
and only $\psi\in L^{2}(\mathbb{R}^{n})$ and $W\psi\in L^{1}(\mathbb{R}^{2n}%
)$. The property follows from the covariance relation $W(\widehat{S}%
\psi)=W\psi\circ S^{-1}$ where $S\in\operatorname*{Sp}(n)$ is the projection
of $\widehat{S}$. (ii) Follows similarly from the translation property
$W(\widehat{T}(z_{0})\psi)=W\psi(z-z_{0})$. (iii) \ Since $\psi$ is continuous
it boundedness follows from $\lim_{z\rightarrow\infty}\psi=0$. Since
$S_{0}(\mathbb{R}^{n})$ is invariant by Fourier transform in view of (i) , we
have $F^{-1}\psi\in S_{0}(\mathbb{R}^{n})$; now $S_{0}(\mathbb{R}^{n})\subset
L^{1}(\mathbb{R}^{n})$ hence $\psi=F(F^{-1}\psi)$ has limit $0$ at infinity in
view of Riemann--Lebesgue's lemma.
\end{proof}

The dual $S_{0}^{\prime}(\mathbb{R}^{n})$ of $S_{0}(\mathbb{R}^{n})$ is
characterized by the following result:

\begin{proposition}
(i) The dual Banach space $S_{0}^{\prime}(\mathbb{R}^{n})$ consists of all
$\psi\in S^{\prime}(\mathbb{R}^{n})$ such that $W(\psi,\phi)\in L^{\infty
}(\mathbb{R}^{2n})$ for one (and hence all) windows $\phi\in S_{0}%
(\mathbb{R}^{n})$; the duality bracket is given by the pairing
\begin{equation}
(\psi,\psi^{\prime})=\int_{\mathbb{R}^{2n}}W(\psi,\phi)(z)\overline
{W(\psi^{\prime},\phi)(z)}dz; \label{dual151}%
\end{equation}
(ii) The formula
\begin{equation}
||\psi||_{\phi,S_{0}^{\prime}(\mathbb{R}^{n})}^{\hbar}=\sup_{z\in
\mathbb{R}^{2n}}|W(\psi,\phi)(z)| \label{dual152}%
\end{equation}
defines a norm on $S_{0}^{\prime}(\mathbb{R}^{n})$. (iii) The Dirac
distribution $\delta$ is in $S_{0}^{\prime}(\mathbb{R}^{n})$; more generally
$\delta x-a)\in S_{0}^{\prime}(\mathbb{R}^{n})$.
\end{proposition}

\begin{proof}
It is based on the fact that $L^{\infty}(\mathbb{R}^{2n})$ is the dual space
of $L^{1}(\mathbb{R}^{2n})$; .see \cite{Gro}, \S 11.3. and \cite{Birkbis}, Ch. 16.
\end{proof}

With the pairing (\ref{dual151}) $(S_{0}(\mathbb{R}^{n}),L^{2}(\mathbb{R}%
^{n}),S_{0}^{\prime}(\mathbb{R}^{n}))$ becomes a Banach Gelfand triple. The
use of the Gelfand triple not only offers a better description of self-adjoint
operators but it also allows a simplification of many proof. Recall that Dirac
already emphasized in his fundamental work \cite{Dirac} the relevance of
\textit{rigged Hilbert spaces} for quantum mechanics.

\section{\label{sec3}Weyl--Heisenberg (Gabor) Multipliers}

\subsection{Weyl--Heisenberg frames}

A lattice in $\mathbb{R}^{2n}\equiv\mathbb{R}_{x}^{n}\times\mathbb{R}_{p}^{n}$
is is a discrete subgroup $\Lambda=M(\mathbb{Z}^{2n})$ of $\mathbb{R}^{2n}$
where $M\in GL(2n,\mathbb{R}).$ One typical example is provided by the matrix
\[
M=%
\begin{pmatrix}
A & 0_{n\times n}\\
0_{n\times n} & B
\end{pmatrix}
\text{ \ },\text{ \ }\det AB\neq0.
\]
In Section \ref{sec3}in which case we have $\Lambda=A\mathbb{Z}^{n}\times
B\mathbb{Z}^{n}$. In many applications one even makes the simpler choice $M=%
\begin{pmatrix}
\alpha I_{n\times n} & 0_{n\times n}\\
0_{n\times n} & \beta I_{n\times n}%
\end{pmatrix}
$ in which case the lattice is just $\Lambda=\alpha\mathbb{Z}^{n}\times
\beta\mathbb{Z}^{n}$.

\begin{definition}
Let $\phi\in L^{2}(\mathbb{R}^{n})$,$||\phi||_{L^{2}}=1$, and a lattice
$\Lambda\subset\mathbb{R}^{2n}$. The set
\[
\mathcal{G}(\phi,\Lambda)=\{\widehat{T}(z)\phi:z\in\Lambda\}
\]
is called a Weyl--Heisenberg (or Gabor) system with lattice $\Lambda$ and
window $\phi$. If there exist constants $a,b>0$ (the frame bounds) such that
for all $\psi\in L^{2}(\mathbb{R}^{n})$ we have the frame condition
\begin{equation}
a||\psi||_{L^{2}}^{2}\leq\sum_{z_{\lambda}\in\Lambda}|(\psi|\widehat{T}%
(z_{\lambda})\phi)|_{L^{2}}|^{2}\leq b||\psi||_{L^{2}}^{2}|^{2} \label{frame}%
\end{equation}
then $\mathcal{G}(\phi,\Lambda)$ is called a Weyl--Heisenberg (or Gabor)
frame. We will use the abbreviation "WH frame" in the text.
\end{definition}

WH frames are generalizations of the notion of basis: if $\mathcal{G}%
(\phi,\Lambda)$ is such a frame, then every $\psi\in L^{2}(\mathbb{R}^{n})$
has a (in general non-unique) expansion%
\begin{equation}
\psi=\sum_{z_{\lambda}\in\Lambda}(\psi|\widehat{T}(z_{\lambda})\phi
)\widehat{T}(z_{\lambda})\phi\label{gaborexp}%
\end{equation}
which means that a square integrable function can be reconstructed from the
knowledge of its orthogonal projections onto the rays $\{c\widehat{T}%
(z_{\lambda})\phi:c\in\mathbb{C}\}$. The following alternative
characterization is almost obvious, and relates Gabor frame theory to the
Weyl--Wigner--Moyal formalism:

\begin{proposition}
Let $\mathcal{G}(\phi,\Lambda)$ be a WH frame. We have
\begin{equation}
\psi=(2\pi\hbar)^{n}\sum_{z_{\lambda}\in\Lambda}\operatorname*{Amb}(\psi
,\phi)(z)\widehat{T}(z_{\lambda})\phi\label{xx1}%
\end{equation}
for every $\psi\in L^{2}(\mathbb{R}^{n})$; when the window $\phi$ is even then%
\begin{equation}
\psi=(\pi\hbar)^{n}\sum_{z_{\lambda}\in\Lambda}W(\psi,\phi)(z)\widehat{T}%
(z_{\lambda})\phi. \label{xx2}%
\end{equation}

\end{proposition}

\begin{proof}
In view of the definition (\ref{ambigter}) of the cross-ambiguity function we
have
\[
(\psi|\widehat{T}(z_{\lambda})\phi)_{L^{2}}=(2\pi\hbar)^{n}\operatorname*{Amb}%
(\psi,\phi)(z)
\]
hence (\ref{xx1}); formula (\ref{xx2}) follows using the relation
(\ref{fouwig2}) between Wigner and ambiguity transforms.
\end{proof}

Note that the definition we give here is slightly different from that usually
given in tome-frequency texts (e.g. \cite{Gro}); for a comparison of both
definitions see Chapter 8 in \cite{Birkbis}.

WH frames enjoy the property of symplectic covariance:

\begin{proposition}
\label{propmeta}If $\mathcal{G}(\mathcal{G}(\phi,\Lambda)$ is a WH frame, the
so is $\mathcal{G}(\mathcal{G}(\widehat{S}\phi,S\Lambda)$ for every
$\widehat{S}\in\operatorname*{Mp}(n)$ covering $S\in\in\operatorname*{Sp}(n)$.
In addition $\mathcal{G}(\mathcal{G}(\phi,\Lambda$ and $\mathcal{G}%
(\mathcal{G}(\widehat{S}\phi,S\Lambda)$ have same frame bounds.
\end{proposition}

\begin{proof}
Assume the frame condition (\ref{frame}) is satisfied; then
\begin{equation}
a||(\widehat{S}\psi||_{L^{2}}^{2}\leq\sum_{z_{\lambda}\in\Lambda
}|((\widehat{S}\psi|\widehat{T}(z_{\lambda})\phi)|_{L^{2}}|^{2}\leq
b||(\widehat{S}\psi||_{L^{2}}^{2}|^{2}%
\end{equation}
that is, since $\widehat{S}$ is unitary,%
\begin{equation}
a||(\psi||_{L^{2}}^{2}\leq\sum_{z_{\lambda}\in\Lambda}|((\psi|\widehat{S}%
^{-1}\widehat{T}(z_{\lambda})\phi)|_{L^{2}}|^{2}\leq b||(\psi||_{L^{2}}%
^{2}|^{2}.
\end{equation}
Using the intertwining formula
\[
\widehat{S}^{-1}\widehat{T}(z)=\widehat{T}(S^{-1}z)\widehat{S}^{1}%
\]
this is equivalent to%
\begin{equation}
a||(\psi||_{L^{2}}^{2}\leq\sum_{z_{\lambda}\in\Lambda}|((\psi|\widehat{T}%
(S^{-1}z_{\lambda})\widehat{S}^{1})\phi)|_{L^{2}}|^{2}\leq b||(\psi||_{L^{2}%
}^{2}|^{2}%
\end{equation}
which can be rewritten%
\begin{equation}
a||(\psi||_{L^{2}}^{2}\leq\sum_{z_{\lambda}\in\Lambda}|((\psi|\widehat{T}%
(S^{-1}z_{\lambda})\widehat{S}^{1})\phi)|_{L^{2}}|^{2}\leq b||(\psi||_{L^{2}%
}^{2}|^{2}.
\end{equation}

\end{proof}

The most basic example of a WH frame is when we chosen as window to be the
standard Gaussian $\phi_{0}^{\hbar}$:

\begin{proposition}
Let $\Lambda_{\alpha\beta}=\alpha\mathbb{Z}^{n}\times\beta\mathbb{Z}^{n}$;
then $\mathcal{G}(\mathcal{G}(\phi_{0}^{\hbar},\Lambda_{\alpha\beta})$ is a WH
frame if and only if $\alpha_{j}\beta_{j}<2\pi\hbar$ for $1\leq j\leq n$.
\end{proposition}

The proof, is reduced to the case $n=1$ using tensor products; see
\cite{Pfander,JGP}.

This result extends to generalized Gaussians as follows:

\begin{corollary}
If $\mathcal{G}(\mathcal{G}(\phi_{0}^{\hbar},\Lambda_{\alpha\beta})$ is a WH
frame, then every $\psi\in L^{2}(\mathbb{R}^{n})$ can be expanded as%
\begin{equation}
\psi=\sum_{z_{\lambda}\in V_{Y}M_{X^{1/2}}\Lambda_{\alpha\beta}}%
(\psi|\widehat{T}(z_{\lambda})\psi_{XY})\widehat{T}(z_{\lambda})\psi_{XY}.
\label{gaborexp2}%
\end{equation}

\end{corollary}

\begin{proof}
It follows from Proposition \ref{propmeta} and formula (\ref{gaborexp}).
\end{proof}

\subsection{Weyl--Heisenberg multipliers}

See Benedetto and Pfander \cite{Benedetto} for a review.

\begin{definition}
Let $a=(a_{\lambda})_{\Lambda}\in\ell_{\mathbb{C}}^{1}(\Lambda)$ be a bounded
sequence of complex numbers and $\mathcal{G}(\mathcal{G}(\phi,\Lambda)$ is a
WH frame. Setting $\phi_{z_{\lambda}}=\widehat{T}(z_{\lambda})\phi$ \ we will
call the operator $\widehat{A}_{\mathrm{GM}}^{\Lambda,\phi}:L^{2}%
(\mathbb{R}^{n})\longrightarrow L^{1}(\mathbb{R}^{n})$ is defined by%
\begin{equation}
\widehat{A}_{\mathrm{GM}}^{\Lambda,\phi}\psi=\sum_{z_{\lambda}\in\Lambda
}a_{\lambda}(\psi|\phi_{z_{\lambda}})\phi_{z_{\lambda}}; \label{GM1}%
\end{equation}
where the Weyl--Heisenberg (or Gabor) multiplier with symbol $a$, lattice
$\Lambda$, and window $\phi$.
\end{definition}

We have (cf. (\ref{xx1}))%
\begin{equation}
\widehat{A}_{\mathrm{GM}}^{\Lambda,\phi}\psi=(2\pi\hbar)^{n}\sum_{z_{\lambda
}\in\Lambda}a_{\lambda}\operatorname*{Amb}(\psi,\phi)(z)\phi_{z_{\lambda}}.
\label{GM2}%
\end{equation}

That $\widehat{A}_{\mathrm{GM}}^{\Lambda,\phi}L^{2}(\mathbb{R}^{n}%
)\longrightarrow L^{1}(\mathbb{R}^{n})$ is easily verified; furthermore:

\begin{theorem}
Suppose $a=(a_{\lambda}(\in\ell_{\mathbb{C}}^{2}(\Lambda)$. Then
$\widehat{A}_{\mathrm{GM}}^{\Lambda,\phi}$ is a Hilbert--Schmidt operator and
is hence compact on $L^{2}(\mathbb{R}^{n})$and we have
\begin{equation}
||\widehat{A}_{\mathrm{GM}}^{\Lambda,\phi}\psi||_{L^{2}}\leq||a||_{\ell
^{\infty}(\Lambda)}b||\psi||_{L^{2}} \label{est1}%
\end{equation}
for every $\psi\in L^{2}(\mathbb{R}^{n})$.
\end{theorem}

\begin{proof}
We have $\widehat{A}_{\mathrm{GM}}^{\Lambda,\phi}=BM_{a}C$ where%
\begin{gather*}
C:\psi\ni L^{2}(\mathbb{R}^{n})\longrightarrow\ell_{\mathbb{C}}^{2}%
(\Lambda)\text{ \ },\text{ \ \ }(\Lambda),C(\psi)=((\psi|\widehat{T}%
(z_{\lambda}))_{\lambda\in\Lambda}\\
M_{a}:\ell_{\mathbb{C}}^{2}(\Lambda)\longrightarrow\ell_{\mathbb{C}}%
^{2}(\Lambda)\text{ \ },\text{ \ \ }M_{a}((\psi|\widehat{T}(z_{\lambda
}))_{\lambda\in\Lambda}=(((a_{\lambda}\psi|\widehat{T}(z_{\lambda}%
))_{\lambda\in\Lambda}\\
B:\ell_{\mathbb{C}}^{2}(\Lambda)\longrightarrow L^{2}(\mathbb{R}^{n})\text{
\ },\text{ \ \ }B(((a_{\lambda}\psi|\widehat{T}(z_{\lambda}))_{\lambda
\in\Lambda}=\widehat{A}_{\mathrm{GM}}^{\Lambda,\phi}\psi.
\end{gather*}
The operator $M_{a}$ is Hilbert--Schmidt and $B,C$ are bounded. The claim
follows. The estimate (\ref{est1}) follows from the definition of
$\widehat{A}_{\mathrm{GM}}^{\Lambda,\phi}$ using the Cauchy--Schwarz inequality.
\end{proof}

WH multipliers also qualify as a density operators for some mixed agates
obtained by considering a generic state $\phi\in L^{2}(\mathbb{R}^{n})$
located at any lattice node with probability $\lambda_{\mu}$:

\begin{proposition}
Let $(\lambda_{\mu})_{\mu\in\Lambda})$ be a sequence of $\geq0$ numbers such
that $\sum_{z_{\lambda}\in\Lambda}\lambda_{\mu}$. The operator%
\[
\widehat{A}_{\mathrm{GM}}^{\Lambda,\phi}=\sum_{z_{\lambda}\in\Lambda}%
\lambda_{\mu}(\psi|\phi_{z_{\lambda}})\phi_{z_{\lambda}})
\]
is a density operator, i.e. it is positive semi-definite and has trace one.
\end{proposition}

\begin{proof}
That $\widehat{A}_{\mathrm{GM}}^{\Lambda,\phi}\geq0$ is clear
\end{proof}

\section{\label{sec4}Toeplitz and anti-Wick Operators}

WH multipliers are a discretized version of the more general notion Toeplitz
operators (which are extensions of anti-Wick) operator \cite{sh87}).The are a
class of particular Weyl operators obtained by smoothing the symbol with an
adequate Wigner transform (that of the standard Gaussian) and are related to a
certain Cohen class \cite{Gro,Birkbis}. For the study of Toeplitz operaors
from the functional analytical point of view see the work of Cordero
\textit{et }al. \cite{Elena1,Elena2,Elena3} and the references therein.

\subsection{Definition and relation with Weyl operators}

Consider the WH multiplier (\ref{GM1}), defined by
\[
\widehat{A}_{\mathrm{GM}}^{\Lambda,\phi}\psi=\sum_{z_{\lambda}\in\Lambda
}a_{\lambda}(\psi|\phi_{z_{\lambda}})\phi_{z_{\lambda}}\text{, \ }(a_{\lambda
})\in\ell_{\mathbb{C}}^{1}(\Lambda);
\]
setting%
\[
a(z)=\sum_{z_{\lambda}\in\Lambda}a_{\lambda}\delta(z-z_{\lambda})
\]
we can rewrites this definition as%
\[
\widehat{A}_{\mathrm{GM}}^{\Lambda,\phi}\psi=\int_{\mathbb{R}^{2n}}%
a(z)(\psi|\phi_{z})_{L^{2}}\phi_{z}dz.
\]

We are following rather closely the description in Shubin \cite{sh87} and our
review in \cite{Birkbis}.

\begin{definition}
Let $a\in S_{0}(\mathbb{R}^{n})$; $||\phi||_{L^{2}}=1$. The Toeplitz operator)
$\widehat{A}_{\text{TO}}^{\phi}=\operatorname*{Op}_{\text{TO }}^{\phi}(a)$
with symbol $a$ and window $\phi$ it $a$ is defined, for $\psi\in
L^{2}(\mathbb{R}^{n})$, by%
\begin{equation}
\widehat{A}_{\text{TO TO }}^{\phi}\psi=\left(  \frac{1}{2\pi\hbar}\right)
^{n}\int_{\mathbb{R}^{2n}}a(z)(\psi|\phi_{z})\phi_{z})_{z}\phi_{z}dz.
\label{defanti8}%
\end{equation}
Equivalently,
\begin{equation}
\widehat{A}_{\text{TO }}^{\phi}\psi=\int_{\mathbb{R}^{2n}}%
a(z)\operatorname*{Amb}(\psi,\phi)(z)\phi_{z}dz. \label{altwick11}%
\end{equation}

\end{definition}

In Dirac's bra-ket notation we can rewrite (\ref{defanti8}( as%
\begin{equation}
\operatorname*{Op}\nolimits_{\text{TO }}^{\phi}(a)=\frac{1}{(2\pi\hbar)^{n}%
}\int_{\mathbb{R}^{2n}}a(z)\,|\phi_{z}\rangle\langle\phi_{z}|\,dz.
\label{Dirac}%
\end{equation}

When $\phi=\phi_{0}$ (the standard Gaussian) we call it the anti-Wick operator
\cite{sh87} with symbol $a$.

Toeplitz operators extend to wider classes of symbols. This follows from the
following results summarize the properties of Toeplitz operators:

\begin{theorem}
\label{ThmTop}Let $\phi\in S_{0}(\mathbb{R}^{n})$. (i) The Toeplitz operator
$\widehat{A}_{\text{TO }}^{\phi}$is the Weyl operator with symbol
\begin{equation}
\widehat{A}_{\text{TO }}^{\phi}=\operatorname*{Op}\nolimits_{\mathrm{Weyl}%
}(a\ast W\phi.). \label{aww8}%
\end{equation}
(ii) When $a=\int1$ then $\widehat{A}_{\text{TO }}^{\phi}$ is the identity:
$_{\text{TO }}^{\phi}(1)=I$; (iiii) $\ \widehat{A}_{\text{TO }}^{\phi}$ is
Hilbert--Schmidt if $a\in L^{2}(\mathbb{R}^{n})$. (iv) If $a\geq0$ then
$\widehat{A}_{\text{TO }}^{\phi}$ is positive semidefinite and self adjoint.
(v) Let $\widehat{S}\in\operatorname*{Mp}(n)$ cover $S\in\operatorname*{Sp}%
(n)$. We have
\begin{equation}
\operatorname*{Op}\nolimits_{\text{TO }}^{\widehat{S}\phi}(a)=\widehat{S}%
\operatorname*{Op}\nolimits_{\text{TO }}^{\phi}(a\circ S))\widehat{S}^{-1}.
\label{opto}%
\end{equation}

\end{theorem}

\begin{proof}
(i) The kernel of $\widehat{A}_{\text{TO }}^{\phi}$ is%
\begin{equation}
K(x,y)=\left(  \frac{1}{2\pi\hbar}\right)  ^{n}\int_{\mathbb{R}^{2n}}%
a(z_{0}))\phi_{z_{0}}(x))\overline{\phi_{z_{0}}(y)}dz_{0} \label{Kernela}%
\end{equation}
and in view of formula /\ref{AK9}) deforming the Weyl symbol we have, setting
$z=(x,p)$.%
\begin{align*}
b(z)  &  =\left(  \frac{1}{2\pi\hbar}\right)  ^{n}\int_{\mathbb{R}^{n}%
}e^{-\frac{i}{\hbar}p\cdot y}K(x+\tfrac{1}{2}y,x-\tfrac{1}{2}y)dy\\
&  =\left(  \frac{1}{2\pi\hbar}\right)  ^{n}\int_{\mathbb{R}^{2n}}%
a(z_{0}))\left(  \int_{\mathbb{R}^{n}}e^{-\frac{i}{\hbar}p\cdot y}\phi_{z_{0}%
}(x+\tfrac{1}{2}y)\overline{\phi_{z_{0}}(x-\tfrac{1}{2}y)}dy\right)  dz_{0}\\
&  =(\int_{\mathbb{R}^{2n}}a(z_{0}))W\phi_{z_{0}}(z)dydz_{0}\\
&  =\int_{\mathbb{R}^{2n}}a(z_{0}))W\phi(z-z_{0}dz_{0}%
\end{align*}
hence formula (\ref{aww8}). (ii) If $a=1$ then
\[
b=1\ast W\phi)=\int_{\mathbb{R}^{2n}}W\phi(z)dz=1.
\]
(iii) Set $b=a\ast W\phi.$ Since $\phi\in S_{0}(\mathbb{R}^{n})$ we have
$W\phi\in L^{1}(\mathbb{R}^{2n})$ hence $a\ast W\phi.\in\in L^{1}%
(\mathbb{R}^{2n})$ and $b$ is square integrable which is a necessary and
sufficient condition for $\widehat{A}_{\text{TO }}^{\phi}$ to be a
Hilbert--Schmidt operator. (iv) That $\widehat{A}_{\text{TO }}^{\phi}$ is elf
adjoint follows from the fact that its Weyl symbol $b$ is real since
$^{n}a\ast W\phi$ is real. To show that $\widehat{A}_{\text{TO }}^{\phi}\geq0$
if $a\geq$ we proceed as follows: for $\psi\in L^{2}(\mathbb{R}^{n})$ we have%
\begin{align*}
(\widehat{A}_{\text{TO }}^{\phi}\psi|\psi)_{L^{2}}  &  =\left(  \frac{1}%
{2\pi\hbar}\right)  ^{n}\int_{\mathbb{R}^{2n}}a(z)(\psi|\phi_{z})\phi
_{z})(\phi_{z}|\psi)_{L^{2}}dz\\
&  =\left(  \frac{1}{2\pi\hbar}\right)  ^{n}\int_{\mathbb{R}^{2n}}%
a(z)|(\psi|\phi_{z})\phi_{z})|^{2}dz\geq0
\end{align*}
hence $\widehat{A}_{\text{TO }}^{\phi}\geq0$.. (v): We have
\begin{align*}
\operatorname*{Op}\nolimits_{\text{TO }}^{\widehat{S}\phi}(a)\psi &  =\left(
\frac{1}{2\pi\hbar}\right)  ^{n}\int_{\mathbb{R}^{2n}}a(z)(\psi|\widehat{T}%
(z)\widehat{S}\phi)_{L^{2}}\widehat{T}(z)\widehat{S}\phi dz\\
&  =\left(  \frac{1}{2\pi\hbar}\right)  ^{n}\int_{\mathbb{R}^{2n}}%
a(z)(\psi|\widehat{S}\widehat{T}(S^{-1}z)\phi)_{L^{2}}\widehat{S}%
\widehat{T}(S^{-1}z)\phi)dz\\
&  =\left(  \frac{1}{2\pi\hbar}\right)  ^{n}\int_{\mathbb{R}^{2n}%
}a(z)(\widehat{S}^{-1}\psi|\widehat{T}(S^{-1}z)\phi)_{L^{2}}\widehat{S}%
\widehat{T}(S^{-1}z)\phi)dz\\
&  =\left(  \frac{1}{2\pi\hbar}\right)  ^{n}\int_{\mathbb{R}^{2n}}%
a(Sz^{\prime})(\widehat{S}^{-1}\psi|\widehat{T}(z^{\prime})\phi)_{L^{2}%
}\widehat{T}(z^{\prime})\widehat{S}\phi)_{L^{2}}\phi dz^{\prime}\\
&  =\widehat{S}\operatorname*{Op}\nolimits_{\text{TO }}^{\phi}(a\circ
S)\widehat{S}^{-1}\psi
\end{align*}
which establishes the covariance (\ref{opto}).
\end{proof}

\subsection{Blob Quantization}

Theorem \ref{ThmTop} has the following consequence:

\begin{corollary}
Let $\widehat{S}_{XY}=\widehat{V}_{Y}\widehat{M}_{X^{1/2},0}$ and
$S_{XY}=V_{Y}M_{X^{1/2},0}$. Let $\phi_{0}$ be the standard Gaussian. We have%
\begin{equation}
\operatorname*{Op}\nolimits_{\text{TO }}^{\psi_{XY}}(a)=\widehat{S}%
_{XY}\operatorname*{Op}\nolimits_{\text{TO }}^{\phi_{0}}(a\circ S_{XY}%
)\widehat{S}_{XY}^{-1}. \label{opwick}%
\end{equation}

\end{corollary}

\begin{proof}
It immediately follows from /\ref{opto}) since $\psi_{XY}=\widehat{S}_{XY}%
\phi_{0}$.
\end{proof}

Observe that, for given symbol $a,$ the operators $\operatorname*{Op}%
\nolimits_{\text{TO }}^{\psi_{XY}}(a)$ are in bijective correspondence with
quantum blobs in view of the bijection (\ref{gamma})
\begin{equation}
\Gamma:\operatorname*{Blob}(n)\longrightarrow\operatorname*{Gauss}(n)
\end{equation}
established in Prop. \ref{Prop2}. This motivates quite naturally the following definition:

\begin{definition}
We will call operator $\operatorname*{Op}\nolimits_{\text{TO }}^{XY}%
a)=\operatorname*{Op}\nolimits_{\text{TO }}^{\psi_{XY}}(a)$ is the "blob
operator" with symbol $a.$associated with quantum blob $S_{XY}B^{2n}%
(\sqrt{\hbar})$.
\end{definition}

The symplectic covariance property (\ref{opto}) considerably simplifies for
blob operators provided on uses the canonical group defined in Section
\ref{sec1}:

\begin{proposition}
Let $(S_{t}^{XY})$ be the canonical group of $\psi_{XY}$. We have, for a$\in
S_{0}(\mathbb{R}^{n})$,%
\begin{equation}
\operatorname*{Op}\nolimits_{\text{TO }}^{\psi_{XY}}(a)=\widehat{S}_{t}%
^{XY}\operatorname*{Op}\nolimits_{\text{TO }}^{\psi_{XY}}(a\circ S_{t}%
^{XY})\widehat{S}_{-t}^{XY}%
\end{equation}

\end{proposition}

\begin{proof}
The symplectic covariance equality (\ref{opto}) becomes%
\begin{equation}
\operatorname*{Op}\nolimits_{\text{TO }XY}^{\widehat{S_{t}}\psi_{XY}%
}(a)=\widehat{S}_{t}^{XY}\operatorname*{Op}\nolimits_{\text{TO }}^{\psi_{XY}%
}(a\circ S_{t}^{XY}))\widehat{S}_{-t}^{XY}.
\end{equation}
We have, by formula (\ref{stxy}),%
\begin{equation}
\widehat{S}_{t}^{XY}\psi_{XY}=\exp\left[  \frac{it}{2\hbar}\hbar
\operatorname*{Tr}(X)\right]  \psi_{XY}%
\end{equation}
hence the result noting that $\psi_{XY}$ and $c\psi_{XY}$ with $|c|=1$ define
the same operator.
\end{proof}

Blob operators are asymptotically close to the Weyl operator with same symbol
in the semiclassical limit $\hbar\rightarrow0$ they reduce asymptotically to
the usual Weyl operators. This is based on the observation hat if
$\widehat{A}_{\text{TO}}^{\psi_{XY}}=\operatorname*{Op}_{\text{TO }}%
^{\psi_{XY}}(a)$ and $\widehat{B}=\operatorname*{Op}\nolimits_{\mathrm{Weyl}%
}(a)$ then $\widehat{A}_{\text{TO}}^{\psi_{XY}}-$ $\widehat{B}$ $\rightarrow0$
when $\hbar\rightarrow0$ . We are not going to make this statement more
precise, but are motivating it by comparing the symbols of both operators. We
begin by noting that in view of formula (\ref{aww8}) in Theorem \ref{ThmTop}
we have $\operatorname*{Op}_{\text{TO }}^{\psi_{XY}}(a)=\operatorname*{Op}%
\nolimits_{\mathrm{Weyl}}(a\ast W\psi_{XY})$. Now $W\psi_{XY}=(\pi\hbar
)^{-n}e^{-Gz\cdot z/\hbar}$ where $G=S^{T}S$ with%
\[
S=%
\begin{pmatrix}
X^{1/2} & 0\\
X^{-1/2}Y & X^{-1/2}%
\end{pmatrix}
\in\operatorname*{Sp}(n).
\]
Now, a simple calculation yields $\lim_{\hbar\rightarrow0}\langle e^{-Gz\cdot
z/\hbar},\theta\rangle=\theta(0)$ for $\theta\in\mathcal{S}(\mathbb{R}^{n})$
hence $\lim_{\hbar\rightarrow0}W\psi_{XY}=\delta$ (Dirac's distribution). It
follows that
\[
\lim_{\hbar\rightarrow0}(a\ast W\psi_{XY}-a)=0
\]
that is, the symbols of both operators are asymptotically identical.

\subsection{Toeplitz operators as density matrices}

Toeplitz operators are very adequate for representing the density matrix (or
operator) of a mixed state. Here is an example (see \ \ \cite{QS}) which
illustrates this fact. Consider $\phi_{z_{0}}=T(z_{0})\phi_{0}$, it is the
ground state of the harmonic oscillator $H(z)=\frac{1}{2}|z-z_{0})^{2}$.
Assuming that $z_{0}$ is not precisely known, and can be any point of phase
space, so the state $\phi_{z_{0}}$ is largely unknown. This lads us to define
the Wigner distribution of this unknown state as being
\begin{equation}
\rho(z)=\int_{\mathbb{R}^{2n}}\mu(z)W(T(z_{0})\phi_{0})(z)/dz \label{tomix}%
\end{equation}
where $\mu$ is a probabiliy distribution. This generalizes in a natural say
the usual situation \cite{QHA} where one deals with a mixed stater, consisting
of a \emph{discrete} "mixture" $(\phi_{j})_{j\in F}$ of states each being
weighted by a probability $\mu_{j}$ and defining the Wigner distribution by
$\rho=\sum_{j\in F}\mu_{j}W\phi_{j\in F}.$. Formula (\ref{tomix}) is thus not
only natural, but much more general. The essential point is to note that in
view of the translation covariance property $W(T(z_{0})\phi_{0})(z)=W(\phi
_{0})(z-z_{0})$ formula (\ref{tomix}) i can be written%
\begin{equation}
\rho(z)=\int_{\mathbb{R}^{2n}}\mu(z)W(\phi_{0})(z-z_{0})/dz=\mu\ast W(\phi
_{0}(z)
\end{equation}
so that
\[
\widehat{\rho}=(2\pi\hbar)^{n}\operatorname*{Op}\nolimits_{\mathrm{Weyl}}%
(\mu\ast W(\phi_{0})=(2\pi\hbar)^{n}\operatorname*{Op}\nolimits_{\text{TO }%
}^{\phi_{0}}(\mu).
\]
This example of course can be generalized without difficulty to more general
situations. In fact:

\begin{theorem}
Let $\rho\in L^{1}(\mathbb{R}^{2n})\cap L^{2}/\mathbb{R}^{2n})$ be a
probability density on $L^{2}(\mathbb{R}^{2n})$. For every window $\phi\in
SS_{0}(\mathbb{R}^{n})$, $||\phi||_{L^{2}}=1$, the Toeplitz operator
\begin{equation}
\widehat{\rho}=(2\pi\hbar)^{n}\operatorname*{Op}\nolimits_{\text{TO }}^{\phi
}(\rho) \label{rhobar}%
\end{equation}
is a density matrix.
\end{theorem}

\begin{proof}
we have $\widehat{\rho}\geq0$ in view of (iv) in Theorem \ref{ThmTop}. Let us
show that $\operatorname*{Tr})\widehat{\rho})=1$. In view of formula
(\ref{aww8})
\[
\widehat{\rho}=(2\pi\hbar)^{n}\operatorname*{Op}\nolimits_{\mathrm{Weyl}}%
(\rho\ast W\phi.)
\]
and hence \cite{Birkbis}\qquad\
\begin{align*}
\operatorname*{Tr}(\widehat{\rho}) &  =\int_{\mathbb{R}^{2n}}(\rho\ast
W\phi)(z)dz\\
&  =(2\pi\hbar)^{n}F(\rho\ast W\phi)(0)\\
&  =(2\pi\hbar)^{2n}F\rho(0)FW\phi)(0).
\end{align*}
\ where $F$ \ is the $2n$-dimensional Fourier transform. Now,
\begin{align*}
F\rho(0) &  =(2\pi\hbar)^{-n}\int_{\mathbb{R}^{2n}}\rho(z)dz=(2\pi\hbar
)^{-n}\\
FW\phi)(0). &  =(2\pi\hbar)\int_{\mathbb{R}^{2n}}W\phi(z)dz=(2\pi\hbar)^{-n}%
\end{align*}
hence $\operatorname*{Tr}(\widehat{\rho})=1$.
\end{proof}

\begin{remark}
The operator $\widehat{\rho}=(2\pi\hbar)^{n}\operatorname*{Op}%
\nolimits_{\text{TO }}^{\phi}(\rho$ being a density matrix, it is a compact
operator hence the spectral theorem tells us that there exists an orthonormal
system $(\phi_{j})$ in $L^{2}/\mathbb{R}^{2n})$ and constants $\lambda_{j}%
\geq0$ summing up to one such that $\widehat{\rho}=\sum_{j}\lambda|\phi
_{j}\rangle\langle\phi_{j}|$. It follows that the Wigner distribution of
$\widehat{\rho}$ is $\sum_{j}\lambda W\phi_{j}$ and hence, comparing with the
theorem above,
\begin{equation}
\rho\ast W\phi=%
{\textstyle\sum_{j}}
\lambda W\phi_{j}.\label{un}%
\end{equation}
This relation seems to be unknown in the literature; it would be interesting
to discuss its physical interpretation.
\end{remark}

\begin{acknowledgement}
This work has been financed by the Austrian Research Foundation FWF
(QuantAustria PAT 2056623)
\end{acknowledgement}

\end{document}